\documentclass[journal,twocolumns]{IEEEtran}

\ifCLASSINFOpdf
\else
\fi

\usepackage{amssymb}
\usepackage{graphicx}
\usepackage{amsmath,amsthm}
\usepackage{mathabx}
\usepackage{mathrsfs}
\usepackage{tikz}
\usetikzlibrary{shapes,arrows}
\usepackage{verbatim}

\newtheorem{theorem}{Theorem}
\newtheorem{lemma}{Lemma}
\newtheorem{remark}{Remark}
\newtheorem{definition}{Definition}

\renewcommand \thesection {\arabic{section}} 

\begin{document}

\tikzstyle{block} = [draw, rectangle, minimum height=3em, minimum width=3em]
\tikzstyle{sum} = [draw, circle, node distance=1cm]
\tikzstyle{input} = [coordinate]
\tikzstyle{output} = [coordinate]
\tikzstyle{pinstyle} = [pin edge={to-,thin,black}]

%
\title{Inversion-based Actuator Fault Estimation from I/O data for Minimum and Non-minimum Phase MIMO Systems}

\author{\IEEEauthorblockN{E.Naderi, K.Khorasani}\\
\IEEEauthorblockA{Department of Electrical and Computer Engineering\\
Concordia University, Montreal, Canada\\kash@ece.concordia.ca }
}


%


\maketitle


\begin{abstract}
We propose a framework for inversion-based estimation of certain categories of faults in discrete-time linear systems. The {fault signal, as an} unknown input{,} is reconstructed from its projections onto two subspaces. One projection is achieved through an algebraic operation, whereas the other is given by a dynamic filter whose poles coincide with the transmission zeros of the system. A feedback is then introduced to stabilize the above filter  as well as to provide an unbiased estimate of the unknown input. {Our solution has two distinctive and practical advantages. First, it represents a unified approach to the problem of inversion of both minimum and non-minimum phase systems as well as systems having transmission zeros on the unit circle. Second, the feedback structure makes the proposed scheme robust to noise}. {We} have shown that the proposed inversion filter is unbiased for certain categories of faults. Finally, we have  illustrated the performance of our proposed methodologies through numerous simulation studies.
\end{abstract}

\renewcommand{\baselinestretch}{0.88} 

\section{Introduction}

{The problem of estimating  faults that  occur in the system actuators and sensors has recently received extensive attention due to advances in the field of fault-tolerant control and growth in demand for higher levels of reliability and autonomy in safety critical systems. A number of approaches have been proposed for fault estimation of dynamical systems, such as unknown input observers (UIO) (\cite{chen2} and \cite{Tan20151048}) and sliding mode observers \cite{alwi2006robust}}. An important category of available solutions are known as \textit{inversion-based} approaches that are addressed in the works such as \cite{szigeti}, \cite{edel1}, \cite{figu}, \cite{kulcsar} and \cite{szabo}. However, these results have one major drawback in common. Specifically, they will fail for non-minimum phase systems. In fact, stable inversion of non-minimum phase systems is an outstanding challenge in any given context associated  with the problem of input reconstruction.

Inversion of linear systems was first systematically treated by Brocket and Mesarovic in \cite{Bro}. The classic references are  \emph{structure algorithm} \cite{silverman}, Sain \& the Massey algorithm \cite{SM}, and the  Moylan algorithm \cite{moylan}. Gillijns \cite{gili} has also proposed a general form of the Sain \& Massey algorithm in which certain free parameters are available that can be adjusted under certain circumstances to obtain a stable inverse system provided that the original system does not have any unstable transmission zeros (that is, minimum phase systems). The inversion problem has also been tackled by more complex methods. Palanthandalam-Madapusi and his colleagues have considered the problem of input reconstruction in several works, however the solutions provided all apply to only minimum-phase systems (\cite{chavan2015delayed}, \cite{palanthandalam2007unbiased} and \cite{kirtikar2009delay}). Flouquet and his colleagues proposed a sliding mode observer for the input reconstruction that is only valid for minimum phase systems \cite{floquet2004sliding}. Marro and Zattoni have proposed a geometric approach \cite{Marro2010815} for state reconstruction of both minimum and non-minimum phase systems. {However, this approach fails for systems that have transmission zeros on the unit circle.}  \\

{In this paper, a solution to fault estimation of linear discrete-time dynamical systems based on a novel inversion-based unknown input reconstruction methodology is proposed. The inversion-based unknown input reconstruction scheme has several practical \textit{advantages} over the available methods in the literature. These advantages further  highlight the main contributions of this work. Specifically, in the available solutions in the literature the system is partitioned into minimum-phase and non-minimum phase parts each of which is separately treated to finally reconstruct the unknown input. Generally, the non-minimum phase part is handled by using the so-called preview-based stable-inversion method. On the other hand, it is well-known that one cannot definitely determine the location of the system zeros due to parameter uncertainties in the model. This fact leads to an important practical issue in case that the system has transmission zeros that are close to the unit circle. In actuality, the system zeros may lie outside, inside or even exactly on the unit circle. Consequently, one cannot successfully apply the available methods to this class of systems. Moreover, preview-based stable-inversion methods are sensitive to noise and they generate large biases in the reconstruction process of the unknown input.}\\

{Our proposed inversion-based method on the other hand overcomes and is void of disadvantages associated with the available methods in the literature. \underline{First}, our proposed methodology can handle both minimum phase and non-minimum phases systems as well as systems having transmission zeros on the unit circle under a \textit{single framework}. Therefore, one does not need to decide on the exact location of the transmission zeros for  application and determination of the most suitable solution. Moreover, to the best of our knowledge, the available solutions in the literature \textit{cannot} cope with the problem of unknown input reconstruction for systems having transmission zeros on the unit circle. \underline{Second},  our solution yields an  estimate of the unknown inputs (i.e., faults)  by   using \underline{only} the system measurements directly (that is, in one single operation)  as it  eliminates the conventional intermediary step of the state estimation process. This is a \textit{significant} improvement and extension from the current practices in the literature for linear systems inversion. Moreover, the reconstruction is rendered through a feedback loop. Both of the proposed schemes  have significantly contributed to and provide robustness subject to presence of  measurement noise. \underline{Third}, our scheme allows relaxation of several restrictive assumptions such as the controllability condition or certain rank conditions that are imposed on the system matrices. However, these advantages come with certain conditions. Specifically, our proposed solution yields unbiased estimation for certain categories of unknown inputs such as step or ramp signals which in fact cover a wide-range of real life phenomena such as faults.}\\

Faults have been modeled in various forms in the literature as either additive  or multiplicative. The proper choice depends on the actual characteristics of the fault. Typically, sensor bias, actuator bias and actuator loss of effectiveness (LOE) are considered as additive faults. Multiplicative fault models are more suitable for representing changes in the system dynamic parameters such as gains and time constants (\cite{patton2013issues}). Moreover,  additive faults are typically considered as LOE that are represented by step-wise or linearly varying (ramp-wise) inputs that are injected to the system. { Our proposed solution perfectly suits}  estimation of step-wise or ramp-wise additive faults that cover a wide range of faults in real life applications.\\

The remainder of the paper is organized as follows. First, the two problems that are considered in this paper are formally stated and defined in Section \ref{sec: problems}. The definitions and notations that are used throughout the paper are provided in Section \ref{sec:notation}. The proposed solution for a stable inversion of linear systems is presented in Section \ref{sec:inversion-method}. The adoption of the proposed inversion method for solving the fault estimation problem  is introduced and developed in Section \ref{sec:fault-estimation}. Finally, numerical simulations and case studies are included in Section \ref{sec:sims}.

\section{Problem Statement}\label{sec: problems}

In this paper, we consider \textit{two} problems as described and formally presented below.
\subsection{Problem 1: Inversion-based input estimation of discrete-time linear systems}

Consider the dynamics of a given linear time-invariant (LTI) discrete-time system that is governed by,

\begin{equation}
{\mathbf{S}:\left\lbrace \begin{array}{l} x(k+1)=Ax(k)+Bu(k)+w(k) \\ y(k)=Cx(k)+Du(k)+v(k)
\end{array} \right.}
\end{equation}
where $x \in \mathbb{R}^n$, $u \in \mathbb{R}^m$ and $y \in \mathbb{R}^l$, where the state $x(t)$ and the input $u(t)$ are assumed to be un-measurable and unavailable. {Moreover, $w (k) \in \mathbb{R}^n$ and $v(k) \in \mathbb{R}^l$ are white noise having zero mean and covariance matrices}
\begin{equation}
{\mathbf{E}[\left[\begin{array}{c} w_i\\v_i\end{array}\right]\left[\begin{array}{cc} w_j^T&v_j^T\end{array}\right]]=\left[\begin{array}{cc} Q&O\\O^T&R\end{array}\right]\delta_{i,j}}
\end{equation}

The main objective that is pursued here is to estimate the unknown sequence ${u(k)}$ from the generated, and the \textit{only} known and available sequence ${y(k)}$ under the following general assumption.
\begin{itemize}
\item  [ ]\textbf{Assumption A}: It is assumed that,
\begin{enumerate}
 \item {The system $\mathbf{S}$ is observable, and}
\item  {At least one of the matrices $B$ or $D$ is full column rank.}
\end{enumerate}
\end{itemize}
In other words, one of the matrices $B$ and $D$ can be rank-deficient or identically zero, but both cannot be simultaneously zero or rank deficient. The other required conditions and assumptions will be given under each result that we will be developing subsequently. We address a solution to the above problem  in Section \ref{sec:inversion-method}.

\subsection{Problem 2: Inversion-based fault estimation of discrete-time linear systems}
Consider a faulty LTI discrete-time system that is given by,

\begin{equation}\label{eq:fault-sys}
\mathbf{S}^f:\left\lbrace \begin{array}{l} x(k+1)=Ax(k)+Bu(k)+Lf(k)+w(k) \\ y(k)=Cx(k)+Du(k)+Ef(k)+v(k)
\end{array} \right.
\end{equation}
where $x \in \mathbb{R}^n$, $u \in \mathbb{R}^m$, $y \in \mathbb{R}^l$ and the input $f \in \mathbb{R}^p$ denotes the fault signal. {Moreover, $w (k) \in \mathbb{R}^n$ and $v(k) \in \mathbb{R}^l$ are white noise having zero mean}. The problem that is considered here is to construct an estimate of the fault signal, i.e. $\hat{f}(k)$, by \textit{only} utilizing the available information from the system, namely $y(k)$ and $u(k)$, under the following assumption.
\begin{itemize}
\item [ ] \textbf{Assumption B}: It is assumed that,
\begin{enumerate}
\item The system $\mathbf{S}^{f}$ is observable, and
\item {At least one of the matrices $L$ or $E$ is full column rank.}
\end{enumerate}
\end{itemize}
The solution to the above problem is discussed and provided subsequently in Section \ref{sec:fault-estimation}.
\section{Notations}\label{sec:notation}

Let us consider the \textit{Rosenbrock System Matrix} that is given by,
\begin{equation}
M_R(z)=\left[\begin{array}{cc}zI-A&B\\-C&D\end{array}\right]
\end{equation}
where if $rank(M_R(z))< n+min (l,m)$, then $z$ is called a \textit{transmission zero} or an \textit{invariant zero} of the system $\mathbf{S}$. Similarly, if  the rank of the following matrix $M_f(z)$ is reduced at a particular value of $z$, the specific zero is designated as the transmission zero of the \textit{fault-to-output} dynamics, where
\begin{equation}
M_f(z)=\left[\begin{array}{cc}zI-A&L\\-C&E\end{array}\right]
\end{equation}

{The vectors $\mathbf{U}(k-2M)$ , $\mathbf{W}(k-2M)$ , $\mathbf{V}(k-2M)$,  $\mathbf{F}(k-2M)$ and $\mathbf{Y}(k-2M)$ that are directly and specifically constructed from the input $u(k)$, process noise $w(k)$, measurement noise $v(k)$, fault $f(k)$ or the output $y(k)$  signals and will be used throughout the paper are defined as follows,}
\begin{equation}\label{eq:u-y-vec-def}
\mathbf{U}(k-2M)=\left[\begin{array}{c}u(k-2M) \\ u(k-2M+1) \\ \vdots \\ u(k-1)
\end{array}  \right]
\end{equation}
{where $M \in \mathbb{N}$ and is selected to be equal or greater than $n$ ($M \geq n$), i.e. the order of the system $\mathbf{S}$. The vectors $\mathbf{F}(k-2M)$, $\mathbf{W}(k-2M)$, $\mathbf{V}(k-2M)$ and $\mathbf{Y}(k-2M)$ are similarly constructed by replacing $u(k)$ in (\ref{eq:u-y-vec-def}) with $f(k)$, $w(k)$, $v(k)$ and $y(k)$, respectively.}

The above input and output vectors satisfy the following relationship,
\begin{multline}\label{eq:lump-sub-rel}
{\mathbf{Y}(k-2M)=\mathbf{C}x(k-2M)+\mathbf{D}\mathbf{U}(k-2M)}\\
{+\mathbf{G}\mathbf{W}(k-2M)+\mathbf{V}(k-2M)}
\end{multline}
where,
\begin{multline}\label{eq:cm-dm-def}
\mathbf{C}=\left(\begin{array}{c}C\\CA\\ \vdots \\ CA^{2M-1}\end{array}\right);  \\ \mathbf{D}=\left(\begin{array}{cccc} D&0& \ldots & 0\\ CB& D& \ldots &0 \\ \vdots & \vdots & \vdots & \vdots \\ CA^{2M-1}B&CA^{2M-2}B&\ldots & D\end{array}\right);\\
\mathbf{G}=\left(\begin{array}{ccccc} 0&0& \ldots & 0&0\\ C& 0& \ldots &0 &0\\ \vdots & \vdots & \vdots & \vdots&\vdots \\ CA^{2M-1}&CA^{2M-2}&\ldots &C& 0\end{array}\right)
\end{multline}

Give a matrix $\mathcal{A}$, then $\mathcal{A}^\perp$, $\mathcal{A}^T$ and $\mathcal{N}(\mathcal{A})$ denote the orthogonal space, the transpose, and the null space of $\mathcal{A}$, respectively. We extensively use the concept of \textit{Moore Penrose pseudo inverse}.  If $\mathcal{A}$ is full row rank, then we denote its \textit{pseudo inverse} by $\mathcal{A}^\dagger$, and compute it by $\mathcal{A}^T(\mathcal{A}\mathcal{A}^T)^{-1}$. Similarly, if $\mathcal{A}$ is full column rank, then we also denote the \textit{pseudo inverse} by $\mathcal{A}^\dagger$, and compute it by $(\mathcal{A}^T\mathcal{A})^{-1}\mathcal{A}^T$. If $\mathcal{A}$ is rank deficient, then we denote the pseudo inverse  by $\mathcal{A}^+$, where $\mathcal{A}^+$ is a matrix that satisfies the following conditions: 1) $\mathcal{A}\mathcal{A}^+\mathcal{A}=\mathcal{A}$, 2) $\mathcal{A}^+\mathcal{A}\mathcal{A}^+=\mathcal{A}^+$, 3) $(\mathcal{A}\mathcal{A}^+)^T=\mathcal{A}\mathcal{A}^+$, and 4) $(\mathcal{A}^+\mathcal{A})^T=\mathcal{A}^+\mathcal{A}$.
If $U\Sigma V^T$ denotes the SVD decomposition of $\mathcal{A}$, then
$\mathcal{A}^+$ is given by $V\Sigma^+U^T$, where $\Sigma^+$ is obtained by
reciprocating each non-zero diagonal element of $\Sigma$.

\section{The proposed inversion-based input estimation of  linear systems}\label{sec:inversion-method}
Our main strategy is to construct $\mathbf{D}\mathbf{U}(k-2M) \in \mathbb{R}^{2Ml}$ by using its projections onto two linearly independent subspaces. First, we identify these subspaces. Next, we will show that the projection of $\mathbf{D}_{}\mathbf{U}_{}(k-2M)$ onto one of these subspaces is directly and simply given by multiplying $\mathbf{Y}_{}(k-2M)$ by a gain. We denote this projection by $\mathbf{U}_{}^{aux}$. Next, we establish an important result that $\mathbf{D}_{}(\mathbf{U}_{}-\mathbf{U}_{}^{aux})$ is zero if the system $\mathbf{S}$ does not have any transmission zeros. Otherwise, computation of the other projection requires that one constructs a dynamical filter. We will identify, specify and characterize this filter and its properties subsequently. Specifically, we will show how the stability condition of this filter is affected by the location of the invariant zeros of the system $\mathbf{S}$.

\subsection{Linear systems with no invariant zeros }
{Let us define the matrix $\mathbf{H}_{}$ as follows,}
\begin{equation}\label{eq:H-org}
\mathbf{H}^T_{}=(\mathbf{C}_{}^T)^\perp
\end{equation}
{Note that since $\mathbf{S}$ is observable as per Assumption A(1), any vector in $\mathbb{R}^{2Ml}$ can be written as a combination of the $\mathbf{C}_{}$ columns and the $\mathbf{H}_{}$ rows. The dot product of the rows of $\mathbf{H}_{}$ with the columns of $\mathbf{D}_{}\mathbf{U}_{}(k-2M)$ is directly given by}
\begin{equation}\label{eq:hy-gu}
{\mathbf{H}_{} \mathbf{D}_{} \mathbf{U}_{}(k-2M) =\mathbf{H}_{} \mathbf{Y}_{}(k-2M)+\mathbf{H}\mathbf{G}\mathbf{W}(k-2M)+\mathbf{H}_{}\mathbf{V}(k-2M)}
\end{equation}

{The matrix $\mathbf{H}_{} \mathbf{D}_{}$ is not a full rank matrix in general. Moreover, the terms $\mathbf{HGW}$ and $\mathbf{V}$ are not known, hence one cannot reconstruct $\mathbf{U}_{}(k-2M)$ from equation (\ref{eq:hy-gu}). To address this challenge, let us determine another input, namely $\mathbf{U}^{aux}_{}(k-2M)$ (designated as the \textit{auxiliary input}), that is obtained by solving the following optimization problem,}
\begin{equation}
\min_{\mathbf{U}^{aux}_{}}\|\mathbf{H}_{} \mathbf{Y}_{}(k-2M)-\mathbf{H}_{}\mathbf{D}_{}\mathbf{U}^{aux}_{}(k-2M)\|
\end{equation}
The solution to the above minimization problem is given by,
\begin{equation}
\mathbf{U}^{aux}_{}(k-2M)=\mathbf{K}_1\mathbf{Y}_{}(k-2M)
\end{equation}
where,
\begin{equation}
\mathbf{K}_1=(\mathbf{H}_{}\mathbf{D}_{})^+\mathbf{H}_{}
\end{equation}
In general, it should be noted that $\mathbf{D}_{}\mathbf{U}^{aux}_{}(k-2M)$ is the construction of $\mathbf{D}_{}\mathbf{U}_{}(k-2M)$ onto the row space of $\mathbf{H}_{}$. Moreover, if the system $\mathbf{S}$ does not have any transmission zeros, then the first $2Ml-n$ rows of $\mathbf{U}_{}(k-2M)$ and $\mathbf{U}_{}^{aux}(k-2M)$ are equal as shown in the following theorem. However, we need to first state the following lemma.
\begin{lemma}\label{lm:square-d-rank}
Let Assumptions A(1) and A(2) hold, $l \geq m$, and $M \geq n$. If the system $\mathbf{S}$ has no transmission zeros, then $rank(\mathbf{D}_{}) \geq 2Mm-n$. The equality holds for square systems, namely when $l=m$.
\end{lemma}
\begin{proof}
Proof is provided in the \textbf{Appendix \ref{app:square-d-rank}}.
\end{proof}
Lemma \ref{lm:square-d-rank} implies that for square systems, as the number of transmission zeros increases, the rank of $\mathbf{D}_{}$ will consequently increase. In other words, $\mathbf{C}_{}$ and $\mathbf{D}_{}$ will have more linearly dependent columns which allows the injection of a nonzero input for zeroing out the output. This fact is also reflected when the problem of decoupling state estimation process from the unknown input is considered.
\begin{theorem}\label{thm:u-uaux}
{Let Assumptions A(1) and A(2) hold, $l \geq m$, and $M \geq n$. If the system $\mathbf{S}$ has no transmission zeros, then at least the first $2Mm-n$ rows of $\mathbb{E}(\mathbf{U}_{}(k-2M)$ - $\mathbf{U}_{}^{aux}(k-2M))$ are zero.}
\end{theorem}
\begin{proof}
Proof is provided in the \textbf{Appendix \ref{app:u-uaux}}.
\end{proof}
Theorem \ref{thm:u-uaux} implies that the unknown input for the system $\mathbf{S}$ having no transmission zeros can be algebraically reconstructed from the measurements.

\subsection{Minimum phase linear systems}

Let us define an \textit{augmented system} $\mathbf{S}^{aug}$ that is governed by,
\begin{equation}\label{eq:s-aug}
\small{
\mathbf{S}^{aug} : \left\lbrace \begin{array}{l} x(k-2M+1)=Ax(k-2M)+B\mathbf{I}_p\mathbf{U}(k-2M)+w(k) \\ \mathbf{Y}(k-2M)=\mathbf{C}x(k-2M)+\mathbf{D}\mathbf{U}_{}(k-2M)\\+\mathbf{G}\mathbf{W}(k-2M)+\mathbf{V}(k-2M)
\end{array} \right.}
\end{equation}
where $\mathbf{I}_p$ is defined according to,
\begin{equation}
\mathbf{I}_p=\left[\begin{array}{cc} \mathbf{I}_{m \times m} & \mathbf{0}_{m \times (2Mm-m)}
\end{array}\right]
\end{equation}
The systems $\mathbf{S}^{aug}$ and $\mathbf{S}$ have the same states, i.e. $x(k)$ subject to $2M$ time delays. Let us also define a dummy state variable $z(k-2M)$ that satisfies the following relationship,
\begin{equation}\label{eq:y-um}
 \mathbf{Y}_{}(k-2M)=\mathbf{C}_{}z(k-2M)+\mathbf{D}_{}\mathbf{U}_{}^{aux}(k-2M)
\end{equation}
The variable $z(k-2M)$ that satisfies the above equation exists since $\mathbf{Y}_{}(k-2M)-\mathbf{D}_{}\mathbf{U}_{}^{aux}(k-2M)$ belongs to the column space of $\mathbf{C}_{}$, and $\mathbf{Y}_{}(k-2M)$ and $\mathbf{U}_{}^{aug}(k-2M)$ are known at each time step. Consequently, $z(k-2M)$ is known and is given by,
\begin{equation}\label{eq:z-comp}
z(k-2M)=\mathbf{C}_{}^\dagger (\mathbf{Y}_{}(k-2M)-\mathbf{D}_{}\mathbf{U}_{}^{aux}(k-2M))
\end{equation}
Note that the variable $z(k)$ is not governed by the dynamics of $x(k)$ except when the system $\mathbf{S}$ does not have any transmission zeros as shown in the proof of Theorem \ref{thm:u-uaux}. In general, $z(k+1)\neq Az(k)+Bu(k)$. In fact the difference between the dynamics of $x(k)$ and $z(k)$ represents the zero dynamics of the system as we will show  subsequently.

Let us define the difference between the two variables as a state error according to,
\begin{equation}\label{eq:e-def}
e(k)=x(k-2M)-z(k-2M)
\end{equation}
Let Assumptions A(1) and A(2) hold, $l \geq m$, and $M \geq n$. The dynamics associated with the state error (\ref{eq:e-def}) is now given by,
{
\begin{multline}\label{eq:state-error-dyn}
e(k+1)=(A-B\mathbf{I}_p\mathbf{D}_{}^+\mathbf{C}_{})e(k) \\ -\left[\begin{array}{ccc}\mathbf{I} & -A& -B\mathbf{I}_p
\end{array}\right]\left[\begin{array}{c}z(k-2M+1)\\z(k-2M)\\\mathbf{U}^{aux}_{}(k-2M)
\end{array}\right]\\
-\mathbf{D}^+(\mathbf{G}\mathbf{W}(k-2M)+w(k-2M)+\mathbf{V}(k-2M))
\end{multline}}
{The above follows given the definition of $e(k)$ as per equation (\ref{eq:e-def}). In other words, we have}
{
\begin{eqnarray}
e(k+1) &=& x(k-2M+1)-z(k-2M+1) \nonumber \\
&=& Ax(k-2M)+B\mathbf{I}_P\mathbf{U}_{}(k-2M)+w(k-2M) \nonumber \\
&-&z(k-2M+1) \nonumber \\
&=& Ae(k)+Az(k-2M)+B\mathbf{I}_P\delta\mathbf{U}_{}(k)\nonumber \\ &+&B\mathbf{I}_P\mathbf{U}_{}^{aux}(k-2M)-z(k-2M+1)+\mathcal{ST} \nonumber \\
&=& (A-B\mathbf{I}_p\mathbf{D}_{}^+\mathbf{C}_{})e(k) \nonumber \\ &-&\left[\begin{array}{ccc}\mathbf{I} & -A& -B\mathbf{I}_p
\end{array}\right]\left[\begin{array}{c}z(k-2M+1)\\z(k-2M)\\\mathbf{U}^{aux}_{}(k-2M)
\end{array}\right]+ \mathcal{ST}. \nonumber \\
\end{eqnarray}}
{where $\mathcal{ST} = -\mathbf{D}^+(\mathbf{G}\mathbf{W}(k-2M)+\mathbf{V}(k-2M)-w(k-2M))$.}

It should be noted that the poles associated with the dynamics (\ref{eq:state-error-dyn}) include the transmission zeros of the system $\mathbf{S}$  for a square system. More specifically, we can state the following result.
\begin{theorem}\label{thm:zeros-of-a-bipkc}
Let Assumptions A(1) and A(2) hold, $l = m$, and $M \geq n$. Let $\mathcal{V}=\{v_i|i=1,..,p\}$ denote the set of the system $\mathbf{S}$ invariant zeros. Let $\mathcal{O}=\{0, \ldots ,0\}$, that contains $n-p$ zeros. The eigenvalues of $(A-B\mathbf{I}_p\mathbf{D}_{}^+\mathbf{C}_{})$ are then given by $\mathcal{V} \cup \mathcal{O}$.
\end{theorem}
\begin{proof}
Proof is provided in the \textbf{Appendix \ref{app:zeros-of-a-bipkc}}.
\end{proof}
Theorem \ref{thm:zeros-of-a-bipkc} links the zero dynamics of the square system $\mathbf{S}$ to the state error dynamics of (\ref{eq:state-error-dyn}). According to this theorem, if a square system $\mathbf{S}$ is minimum phase, then the state error dynamics (\ref{eq:state-error-dyn}) will be stable. This statement is not generally true for non-square systems, since the state error dynamics (\ref{eq:state-error-dyn}) may have unstable pole(s) even for non-square minimum phase systems.

The state error dynamics is associated with the difference between $\mathbf{U}_{}(k-2M)$ and $\mathbf{U}_{}^{aux}(k-2M)$ as follows. If we define,
\[\delta \mathbf{U}_{}(k)=\mathbf{U}_{}(k-2M)-\mathbf{U}_{}^{aux}(k-2M)\]
and subtract equation (\ref{eq:y-um}) from the measurement equation of the system $\mathbf{S}^{aug}$, one will obtain,
\begin{equation}\label{eq:du-ce}
{\mathbf{D}_{}\delta \mathbf{U}_{}(k)=-\mathbf{C}_{}e(k)- \mathbf{G}\mathbf{W}(k-2M)-\mathbf{V}(k-2M)}
\end{equation}
The dynamics (\ref{eq:state-error-dyn}) along with equation (\ref{eq:du-ce}) can be used to construct an inverse filter for a square minimum phase system as follows. Towards this end, we first provide a definition and present a lemma.
\begin{definition}
{Consider a sequence $u(k)$. We let $\hat{u}(k)$ denote an unbiased estimate of $u(k)$ if $\hat{u}(k) \rightarrow z^{-q}\mathbb{E}(u(k))$ as $k \rightarrow \infty$, where $q \in \mathbb{N}$. Otherwise, it will be designated as a biased estimate of $u(k)$.}
\end{definition}
\begin{lemma}\label{lm:ipd}
Let Assumptions A(1) and A(2) hold, $l \geq m$, and $M \geq n$.  Then it follows that $\mathbf{I}_P.\mathcal{N}(\mathbf{D}_{})=0$.
\end{lemma}
\begin{proof}
Proof is provided in the \textbf{Appendix \ref{app:lm-ipd}}.
\end{proof}
We are now in a position to state our next main result.
\begin{theorem}\label{thm:mp-filter}
Let Assumption A(1) and A(2) hold, $l = m$, and $M \geq n$. If the system $\mathbf{S}$ is minimum phase, then the unbiased estimate of the unknown input $u(k-2M)$ is governed by the filter dynamics,
\begin{equation}\label{eq:mp-filter}
\mathbf{S}^{inv} : \left\lbrace \begin{array}{l} \hat{e}(k+1)=(A-B\mathbf{I}_p\mathbf{D}_{}^+\mathbf{C}_{})\hat{e}(k)-B_F\mathcal{U}(k-2M) \\
\hat{\mathbf{U}}_{}(k)=-\mathbf{D}_{}^+\mathbf{C}_{}\hat{e}(k)+\mathbf{U}^{aux}_{}(k-2M) \\
\hat{u}(k)=\mathbf{I}_p \hat{\mathbf{U}}_{}(k)
\end{array} \right.
\end{equation}
where,
\begin{equation}
B_F=\left[\begin{array}{ccc}\mathbf{I} & -A& -B\mathbf{I}_p
\end{array}\right]
\end{equation}
\begin{equation}\label{eq:u-mathcal}
\mathcal{U}(k-2M)=\left[\begin{array}{c}z(k-2M+1)\\z(k-2M)\\\mathbf{U}^{aux}_{}(k-2M)
\end{array}\right].
\end{equation}
where the state $z(k)$ at each time step is given by equation (\ref{eq:z-comp}).
\end{theorem}
\begin{proof}
Proof is provided in the \textbf{Appendix \ref{app:mp-filter}}.
\end{proof}
\subsection{Non-minimum phase systems}\label{subsec:nmp-sys}
It should be noted that one cannot use Theorem \ref{thm:mp-filter} for non-minimum phase and/or non-square systems as well as systems with transmission zeros on the unit circle.

Consequently, below we will derive the dynamics associated with $\delta \mathbf{U}_{}(k)$ and attempt to stabilize it to ensure a zero tracking error. Let us define,
\begin{equation}
\eta(k)=\mathbf{D}_{}\delta \mathbf{U}_{}(k)
\end{equation}
It now follows that the dynamics of $\eta(k)$ is governed by,
\begin{equation} \label{eq:eta-dyn}
{\eta(k+1)=\tilde{A}\eta(k)+\mathbf{C}_{}B_F\mathcal{U}(k-2M)+\mathcal{ST'}}
\end{equation}
where
\begin{equation}
\tilde{A}=\mathbf{C}_{}(A-B\mathbf{I}_P \mathbf{D}_{}^+\mathbf{C}_{})\mathbf{C}_{}^\dagger.
\end{equation}
and
\begin{equation}
  {  \mathcal{ST'}=\mathbf{C}\mathbf{D}^+(\mathbf{G}\mathbf{W}(k-2M)-\mathbf{V}(k-2M))}
\end{equation}
This follows by multiplying both sides of equation (\ref{eq:state-error-dyn}) by $\mathbf{C}_{}$ and then replacing $\mathbf{C}_{}e(k)$ by equation (\ref{eq:du-ce}), to yield the result.

In order to obtain a stable filter for non-minimum phase systems that is applicable to both square and non-square systems, we rotate both $\mathbf{C}_{}$ and $\mathbf{H}_{}$ through a rotation matrix $\mathbf{R} \in \mathbb{R}^{2M \times 2M}$ about an \textit{arbitrary axis} as follows,
\begin{equation}
\mathbf{C}_{}^{new}=\mathbf{R}\mathbf{C}_{}
\end{equation}
\begin{equation}
\mathbf{H}_{}^{new}=(\mathbf{R}\mathbf{H}_{}^T)^T
\end{equation}
A square matrix is said to be a rotation matrix if $\mathbf{R}\mathbf{R}^T=\mathbf{R}^T\mathbf{R}=\mathbf{I}$ and $\|\mathbf{R}\|=1$. This operation represents a \textit{similarity transformation} for the following system \footnote{Note that the system matrices of $\mathbf{S}^\eta$, i.e. $(\tilde{A},\mathbf{C}_{}B_F,\mathbf{H}_{})$ after applying the similarity transformation of the matrix $\mathbf{R}$ is represented by $(\mathbf{R}\tilde{A}\mathbf{R}^T,\mathbf{C}_{}^{new}B_F, \mathbf{H}_{}^{new})$. },
\begin{equation}\label{eq:s-eta}
{\mathbf{S}^{\eta}:\left\lbrace \begin{array}{l} \eta(k+1)=\tilde{A}\eta(k)+\mathbf{C}_{}B_F\mathcal{U}(k-2M)+\mathcal{ST'} \\ \mathcal{H}(k)=\mathbf{H}_{}\eta(k)
\end{array} \right.}
\end{equation}
{Note that $\mathbb{E}(\mathcal{H}(k)) \equiv 0$  since,}
{
\begin{eqnarray}
\mathcal{H}(k)&=&\mathbf{H}_{}\eta(k)= \mathbf{H}_{}\mathbf{D}_{}\delta \mathbf{U}_{}(k) \nonumber \\
&=&\mathbf{H}_{}\mathbf{D}_{}(\mathbf{U}_{}(k-2M)-\mathbf{U}_{}^{aux}(k-2M)) \nonumber \\
&=& \mathbf{H}_{}(\mathbf{Y}_{}(k-2M)-\mathbf{G}\mathbf{W}(k-2M)-\mathbf{V}(k-2M) \nonumber \\
&-&\mathbf{Y}_{}(k-2M))
\end{eqnarray}}
Therefore,

\[ {\mathbb{E}(\mathcal{H}(k)) = 0} \]
Therefore, if the system $\mathbf{S}$ has any transmission zeros, then the difference between the real input and the auxiliary input serves as the output-zeroing input of the system (\ref{eq:s-eta}). One may have suggested now to use the feedback from $\mathcal{H}(k)$ to stabilize the system  $\mathbf{S}^{\eta}$. However, clearly the system $\mathbf{S}^{\eta}$ is neither controllable nor observable.

Therefore, we now instead define $\hat{\eta}(k)$ to be governed by,
\begin{multline}\label{eq:eta-hat}
\hat{\eta}(k+1)=(\mathbf{P}_c^{new}\tilde{A}+\mathbf{P}_h^{new}+\mathbf{K}_2 \mathbf{P}_h)\hat{\eta}(k) \\ +\mathbf{P}_c^{new}\mathbf{C}_{}B_F\mathcal{U}(k-2M)
\end{multline}
where,

\[\mathbf{P}_h^{new}=\mathbf{H}^{new^T}_{}(\mathbf{H}_{}^{new}\mathbf{H}^{new^T}_{})^{-1}\mathbf{H}_{}^{new}\]
\[\mathbf{P}_c^{new}=\mathbf{C}^{new}_{}(\mathbf{C}_{}^{new^T}\mathbf{C}^{new}_{})^{-1}\mathbf{C}_{}^{new^T}\]
with $\mathbf{K}_2$ chosen such that all the eigenvalues of $(\mathbf{P}_c^{new}\tilde{A}+\mathbf{P}_h^{new}+\mathbf{K}_2 \mathbf{P}_h)$ lie inside the unit circle.
 {Note that if the unknown input is a step function, then $\mathbb{E}(\eta(k)-\hat{\eta}(k)) \rightarrow 0$ as $k \rightarrow \infty$} \footnote{If one could design a filter in the form of $\hat{\eta}(k+1)=(\mathbf{P}_c^{new}\tilde{A}+\mathbf{P}_h^{new}\tilde{A}+\mathbf{K}_2 \mathbf{P}_h)\hat{\eta}(k)+\mathbf{P}_c^{new}\mathbf{C}_{}B_F\mathcal{U}$, then one would have an unbiased estimation of all types of inputs, however, this filter and similar ones would unfortunately be neither controllable nor observable.}.

In order to establish the above claim, first, we discuss the stabilization of the filter (\ref{eq:eta-hat}) through selection of $\mathbf{K}_2$ and then address its tracking error behavior and performance.

It can be easily concluded that the stabilization of the filter (\ref{eq:eta-hat}) by the gain $\mathbf{K}_2$ is possible if and only if the pair $(\mathbf{P}_c^{new}\tilde{A}+\mathbf{P}_h^{new},-\mathbf{P}_h)$ is observable, which provides an explicit criterion for selection of the rotation matrix $\mathbf{R}$. However, certain care should be exercised in selection of $\mathbf{R}$ as pointed out in the following two remarks.

\begin{remark}\label{rm:obs1}
If $\mathbf{R}$ is selected such that the column space of $\mathbf{C}_{}^{new}$ coincides with the column space of $\mathbf{C}_{}$ (or equivalently the row space of $\mathbf{H}_{}^{new}$ coincides with the row space of $\mathbf{H}_{}$), then the pair $(\mathbf{P}_c^{new}\tilde{A}+\mathbf{P}_h^{new},-\mathbf{P}_h)$ will \underline{not} be observable since (a) $\mathbf{P}_h^{new}=\mathbf{P}_h$, (b) $\mathbf{P}_h(\mathbf{P}_c^{new}\tilde{A}+\mathbf{P}_h^{new})=\mathbf{P}_h^{new}$, and (c) $\mathbf{P}_h$ is column rank deficient. Hence, the observability matrix will be rank deficient.
\end{remark}
\begin{remark}\label{rm:obs2}
If $\mathbf{R}$ is selected such that the column space of $\mathbf{C}_{}^{new}$ coincides with the row space of $\mathbf{H}_{}$, then the pair $(\mathbf{P}_c^{new}\tilde{A}+\mathbf{P}_h^{new},-\mathbf{P}_h)$ will not be observable since $\mathbf{P}_h(\mathbf{P}_c^{new}\tilde{A}+\mathbf{P}_h^{new})=0$, and therefore the observability matrix will be rank deficient.
\end{remark}
Geometrically speaking, for a SISO system having a single state, Remarks \ref{rm:obs1} and \ref{rm:obs2} imply that $\mathbf{R}$ should not be a matrix resulting in a rotation of $\frac{q\pi}{2}$, $q \in \mathbb{Z}$, about the axis passing through origin and should be perpendicular to both $\mathbf{C}_{}$ and $\mathbf{H}_{}$. Otherwise, for example for a rotation angel of $\frac{\pi}{2}$, the column space of $\mathbf{C}_{}^{new}$ will coincide with the row space of $\mathbf{H}_{}$. All other $\mathbf{R}$s except those excluded in Remarks \ref{rm:obs1} and \ref{rm:obs2} will yield an observable pair $(\mathbf{P}_c^{new}\tilde{A}+\mathbf{P}_h^{new},-\mathbf{P}_h)$. However, the closer the rotation angel is to $\frac{q\pi}{2}$, a higher gain $\mathbf{K}_2$ will be required to stabilize the system. This will be numerically illustrated in the simulation case studies in Section \ref{sec:sims}.

Moreover, if a square system has one or more transmission zeros \textit{exactly equal} to 1 (with \textit{no}  other transmission zeros on the unit circle), then there will exist no $\mathbf{R}$ such that the pair $(\mathbf{P}_c^{new}\tilde{A}+\mathbf{P}_h^{new},-\mathbf{P}_h)$ is observable. We can now state the following result.
\begin{lemma}\label{rm:obs3}
If a square system $\mathbf{S}$ has a transmission zero exactly equal to 1 ($z=1$), then the pair $(\mathbf{P}_c^{new}\tilde{A}+\mathbf{P}_h^{new},-\mathbf{P}_h)$  will not be observable for any selection of the rotation matrix $\mathbf{R}$.
\end{lemma}
\begin{proof}\label{pf:obs3}
Proof is provided in the \textbf{Appendix \ref{app:obs3}}.
\end{proof}
If a square system $\mathbf{S}$ has transmission zeros on the unit circle except at $z=1$, then every $\mathbf{R}$ except those stated in Remarks \ref{rm:obs1} and \ref{rm:obs2} yield an observable pair $(\mathbf{P}_c^{new}\tilde{A}+\mathbf{P}_h^{new},-\mathbf{P}_h)$. Non-square systems rarely have transmission zeros (\cite{DAVISON1974643}), therefore it is less likely to have a transmission zero that is equal to 1, or in general on the unit circle. If so then a matrix $\mathbf{R}$ may or may not exist.

Once the observability condition is satisfied, it is straightforward to determine $\mathbf{K}_2$ by using the Ackerman's method to place the system poles at desired locations. The significance of our proposed solution can be appreciated by the fact that the designed feedback not only stabilizes the system for both minimum and non-minimum phase systems in general, but also it provides an unbiased estimate of the unknown step input as stated in the following theorem.
\begin{theorem}\label{thm:nmp-filter}
Let the Assumptions A(1) and A(2) hold, $l \geq m$, and $M \geq n$. If the unknown input is a step function, and if there exists an $\mathbf{R}$ such that the pair $(\mathbf{P}_c^{new}\tilde{A}+\mathbf{P}_h^{new},-\mathbf{P}_h)$ is observable, and $\mathbf{K}_2$ is chosen such that all the eigenvalues of $\mathbf{P}_c^{new}\tilde{A}+\mathbf{P}_h^{new}+\mathbf{K}_2\mathbf{P}_h$ lie inside the unit circle,	then an unbiased estimate of the unknown input $u(k-2M)$ is given by,

\begin{equation}\label{eq:nmp-filter}
\mathbf{S}^{inv}_{stp} : \left\lbrace \begin{array}{l} \hat{\eta}(k+1)=(\mathbf{P}_c^{new}\tilde{A}+\mathbf{P}_h^{new}+\mathbf{K}_2 \mathbf{P}_h)\hat{\eta}(k) \\+\mathbf{P}_c^{new}\mathbf{C}_{}B_F\mathcal{U}(k-2M)\\
\hat{\mathbf{U}}_{}(k)=\mathbf{D}_{}^+\hat{\eta}(k)+\mathbf{U}^{aux}_{}(k-2M) \\
\hat{u}(k)=\mathbf{I}_p \hat{\mathbf{U}}_{}(k)
\end{array} \right.
\end{equation}
\end{theorem}
\begin{proof}
Proof is provided in the \textbf{Appendix \ref{app:nmp-filter}}.
\end{proof}
Note that in contrast to the filter (\ref{eq:mp-filter}), which is limited to only square and minimum phase systems, the filter (\ref{eq:nmp-filter}) is a general solution for both minimum and non-minimum phase systems of any size that satisfies $l \geq m$ \footnote{One can obtain similar results with $\hat{\eta}(k+1)=(\mathbf{P}_h^{new}\tilde{A}+\mathbf{P}_c^{new}+\mathbf{K}_2 \mathbf{P}_h)\hat{\eta}(k)+\mathbf{P}_h^{new}\mathbf{C}_{}B_F\mathcal{U}$ due to symmetrical properties of the rotation matrix.}.  Moreover, it \textit{can handle} systems that have transmission zeros on the unit circle.

By a close inspection of the proof of Theorem \ref{thm:nmp-filter} it follows that the strategy for constructing a stable and unbiased inversion filter for an \textit{unknown ramp} as well as \textit{step input functions} can be developed. The strategy for the ramp input is to specifically construct a filter that results in increasing the \textit{type} of the error dynamics to diminish the steady state error.  Based on the above observation, the following theorem can now be stated.
\begin{theorem}\label{thm:nmp-filter-rmp}
Let Assumptions A(1) and A(2) hold, $l \geq m$, and $M \geq n$. If the unknown input is a ramp function, and if there exists a rotation matrix $\mathbf{R}$ such that the pair $(\mathbf{P}_h^{new}\tilde{A}^2-2\mathbf{P}_h^{new}\tilde{A}+\tilde{A}+\mathbf{P}_h^{new},-\mathbf{P}_h)$ is observable, and $\mathbf{K}_2$ is chosen such that all the eigenvalues of $\mathbf{P}_h^{new}\tilde{A}^2-2\mathbf{P}_h^{new}\tilde{A}+\tilde{A}+\mathbf{P}_h^{new}+\mathbf{K}_2\mathbf{P}_h$ lie inside the unit circle, then an unbiased estimate of the unknown input $u(k-2M)$ is given by,
\begin{equation}\label{eq:nmp-filter-rmp}
\mathbf{S}^{inv}_{rmp} : \left\lbrace \begin{array}{l} \hat{\eta}(k+1)=\\ (\mathbf{P}_h^{new}\tilde{A}^2-2\mathbf{P}_h^{new}\tilde{A}+\tilde{A}+\mathbf{P}_h^{new}+\mathbf{K}_2 \mathbf{P}_h)\hat{\eta}(k) \\ +\Gamma(k-2M)\\
\hat{\mathbf{U}}_{}(k)=\mathbf{D}_{}^+\hat{\eta}(k)+\mathbf{U}^{aux}_{}(k-2M) \\
\hat{u}(k)=\mathbf{I}_p \hat{\mathbf{U}}_{}(k)

\end{array} \right.
\end{equation}
where,
\begin{multline}
\Gamma(k-2M)=\mathbf{P}_h^{new}\mathbf{C}_{}B_F\mathcal{U}(k-2M+1) \\ +\left(\mathbf{P}_h^{new}\tilde{A}-2\mathbf{P}_h^{new}+\mathbf{I}\right)\mathbf{C}_{}B_F\mathcal{U}(k-2M).
\end{multline}

\end{theorem}

\begin{proof}
Proof is provided in the \textbf{Appendix \ref{app:nmp-filter-rmp}}.
\end{proof}
It is interesting to note that the filter (\ref{eq:nmp-filter-rmp}) \underline{cannot} be obtained through standard and basic mathematical operations (such as a similarity transformation) from the filter (\ref{eq:nmp-filter}) or vice versa. This concludes our proposed general solution to inversion of discrete-time linear systems.

To summarize, the unknown input was reconstructed from its projection onto the column space of $\mathbf{C}_{}$ and the row space of $\mathbf{H}_{}$. The projection on the row space of $\mathbf{H}_{}$ is simply given by equation (\ref{eq:hy-gu}), however, the projection on $\mathbf{C}_{}$ is indirectly obtained from the reconstruction of $\mathbf{D}_{}\delta\mathbf{U}$. The term $\mathbf{D}_{}\delta\mathbf{U}$ has this important property that it is orthogonal to the subspace that is spanned by the rows of $\mathbf{H}_{}$.

Yet, two important issues are associated with our proposed technique. First, the construction of $\mathbf{D}_{}\delta\mathbf{U}$ is an unstable process for non-minimum phase systems. Second, computation of $\delta\mathbf{U}$ requires the inverse of $\mathbf{D}_{}$, which is a non-square and rank-deficient matrix under most circumstances.

To address the first issue we have proposed a novel technique in which the column space of $\mathbf{C}_{}$ and the row space of $\mathbf{H}_{}$ are transformed through a rotation matrix about an arbitrary axis, followed by introducing a feedback that not only stabilizes, but also eliminates the steady state error of the inverse filter. To address the second issue, Lemma \ref{lm:ipd} is introduced that is always satisfied for minimal systems with $l \geq m$, even if $\mathbf{D}_{}$ is rank deficient.

In the next section, we provide a solution to our \underline{Problem 2} that was introduced in Section II.

\section{The Proposed Inversion-Based Fault Estimation for Non-minimum Phase Fault to Output Systems}\label{sec:fault-estimation}

One of the most important applications of  system inversion is in the problem of fault estimation. A solution to this problem is  essential for any successful fault-tolerant control scheme and reliable operation of most engineering systems. In this section, we  show that our proposed system inversion approach can be easily adopted for fault estimation purposes. The advantage of our methodology is that the unknown fault input is directly reconstructed from only the system measurements \textit{without} requiring any \textit{a priori} estimate of the system states. Moreover, it can handle transmission zeros everywhere on the complex plan even on the unit circle.

We follow a similar procedure that was proposed in the previous section with the difference that now in the system $\mathbf{S}^f$, $u(k)$ is assumed to be known and the unknown input, which is the injected fault signal, is now designated as $f(k)$.

Therefore, let us define the vector $\mathbf{F}^{aux}_{}$ as follows,
\begin{equation}
\mathbf{F}^{aux}_{}(k-2M)=\mathbf{K}_1^f(\mathbf{Y}_{}(k-2M)-\mathbf{D}_{}\mathbf{U}_{}(k-2M))
\end{equation}
where $\mathbf{K}_1^f$ is given by,
\begin{equation}
\mathbf{K}_1^f=(\mathbf{H}_{}\mathbf{E}_{})^+\mathbf{H}_{}
\end{equation}and
\begin{equation}\label{eq:em-def}
\mathbf{E}_{}=\left(\begin{array}{cccc} E&0& \ldots & 0\\ CL& E& \ldots &0 \\ \vdots & \vdots & \vdots & \vdots \\ CA^{2M-1}L&CA^{2M-2}L&\ldots & L\end{array}\right)
\end{equation}
According to Theorem \ref{thm:u-uaux}, $\mathbf{F}^{aux}_{}(k-2M)$ represents a construction of $\mathbf{F}_{}(k-2M)$ if the fault-to-output dynamics has no transmission zeros. For the general case, we define a dummy state variable $z^f(k-2M)$ that satisfies the following relationship,
{
\begin{multline}\label{eq:y-z-f}
 \mathbf{Y}_{}(k-2M)=\mathbf{C}_{}z^f(k-2M)+\mathbf{D}_{}\mathbf{U}_{}(k-2M)\\
 +\mathbf{E}_{}\mathbf{F}_{}^{aux}(k-2M)+\mathbf{G}\mathbf{W}(k-2M)+\mathbf{V}(k-2M)
\end{multline}}
Moreover,  we define,
\begin{equation}
\eta^f(k)=\mathbf{E}_{}\delta \mathbf{F}_{}(k)=\mathbf{E}_{}(\mathbf{F}_{}(k-2M)-\mathbf{F}_{}^{aux}(k-2M))\end{equation}
Therefore, the dynamics associated with $\eta^f(k)$ is now governed by
\begin{equation}\label{eq:etaf-dyn}
\eta^f(k+1)=\tilde{A}^f\eta^f(k)+\mathbf{C}_{}B_F^f\mathcal{U}^f(k-2M).
\end{equation}
where,
\begin{equation}
\tilde{A}^f=\mathbf{C}_{}(A-L\mathbf{I}_p^f\mathbf{E}_{}^+\mathbf{C}_{})\mathbf{C}_{}^\dagger
\end{equation}
\begin{equation}
B_F^f=\left[\begin{array}{cccc}\mathbf{I} & -A& -L\mathbf{I}_p^f &-B\mathbf{I}_p
\end{array}\right]
\end{equation}
\begin{equation}
\mathcal{U}^f(k-2M)=\left[\begin{array}{c}z^f(k-2M+1)\\z^f(k-2M)\\\mathbf{F}^{aux}_{}(k-2M)\\\mathbf{U}_{}(k-2M)
\end{array}\right]
\end{equation}
and
\begin{equation}
\mathbf{I}_p^f=\left[\begin{array}{cc} \mathbf{I}_{p \times p} & \mathbf{0}_{p \times (2Mp-p)}
\end{array}\right]
\end{equation}

Note that as compared to equation (\ref{eq:eta-dyn}), the additional known information $\mathbf{U}_{}(k-2M)$ appears in $\mathcal{U}^f(k-2M)$. The dynamics of the system (\ref{eq:etaf-dyn}) is unstable if the fault-to-output dynamics has transmission zeros outside or on the unit circle. On the other hand, a close examination of the dynamics (\ref{eq:etaf-dyn}) reveals that it is quite similar to the dynamics that is governed by (\ref{eq:eta-dyn}). Therefore, the same strategy that was described in the previous section can now be applied here. Specifically, we can conclude the following result.

\begin{theorem}\label{thm:fault-est}
Let Assumptions B(1) and B(2) hold, $l \geq p$, and $M \geq n$. If the fault signal is a step loss of effectiveness (LOE) function, and there exists a rotation matrix $\mathbf{R}$ such that the pair $(\mathbf{P}_c^{new}\tilde{A}^f+\mathbf{P}_h^{new},-\mathbf{P}_h)$ is observable, and $\mathbf{K}_2^f$ is chosen such that all the eigenvalues of $\mathbf{P}_c^{new}\tilde{A}^f+\mathbf{P}_h^{new}+\mathbf{K}_2\mathbf{P}_h$ lie inside the unit circle, then an unbiased estimate of the fault vector $f(k-2M)$ is given by,
\begin{equation}\label{eq:nmp-filter-fault}
\mathbf{S}^{inv,f}_{stp} : \left\lbrace \begin{array}{l} \hat{\eta}^f(k+1)=\\(\mathbf{P}_c^{new}\tilde{A}^f+\mathbf{P}_h^{new}+\mathbf{K}_2^f \mathbf{P}_h)\hat{\eta}^f(k)\\+\mathbf{P}_c^{new}\mathbf{C}_{}B_F^f\mathcal{U}^f(k-2M)\\
\hat{\mathbf{F}}_{}(k)=\mathbf{E}_{}^+\hat{\eta}^f(k)+\mathbf{F}^{aux}_{}(k-2M) \\
\hat{f}(k)=\mathbf{I}_p \hat{\mathbf{F}}_{}(k)
\end{array} \right.
\end{equation}
\end{theorem}

\begin{proof}
Proof is not included, since it is similar to the proof of Theorem \ref{thm:nmp-filter}.
\end{proof}
One can also  establish a result that is similar to Theorem \ref{thm:nmp-filter-rmp} for the case when the fault signal is a ramp (drift) loss of effectiveness (LOE) function.
The details are not included for brevity.

This now concludes our proposed methodology for estimation of the loss of effectiveness faults for systems having transmission zeros anywhere on the complex plan. In the next section, we provide illustrative simulations that demonstrate the merits and capabilities of our proposed methodologies.

\section{Four Case Studies}\label{sec:sims}
{For the \textit{first simulation case study} } consider a first order non-minimum phase SISO system that is governed by \footnote{For all simulations in this section, we set $M=n$.},

\begin{equation}\label{sim:sys1}
  S : \left\lbrace \begin{array}{l} x(k+1)=0.5x(k)+u(k) \\
  y(k)=-x(k)+u(k)
   \end{array} \right.
\end{equation}

The transfer function of the system is given by,
\begin{equation}
  G(z)=\frac{z-1.5}{z-0.5}
\end{equation}

For the above system, it follows that $\mathbf{C}_{}=\left[\begin{array}{c}-1\\-0.5\end{array}\right]$, and consequently $\mathbf{H}_{}=\left[\begin{array}{cc}-0.45&0.90\end{array}\right]$. The rotation matrix $\mathbf{R}$ is given by,
\[\mathbf{R}(\theta)=\left[\begin{array}{cc}cos(\theta)&-sin(\theta)\\sin(\theta)& cos(\theta)\end{array}\right]\]
According to Remarks \ref{rm:obs1} and \ref{rm:obs2}, the pair $(\mathbf{P}_c^{new}\tilde{A}+\mathbf{P}_h^{new},-\mathbf{P}_h)$ is not observable for $\theta =\frac{q\pi}{2},q \in \mathbb{Z}$, where,
\[\tilde{A}=\left[\begin{array}{cc} 1.2 &0.6\\0.6& 0.3\end{array}\right];\mathbf{P}_h=\left[\begin{array}{cc} 0.2 &-0.4\\-0.4& 0.8\end{array}\right]\]
\[ \mathbf{P}_c=\left[\begin{array}{cc} 0.8 &0.4\\0.4& 0.2\end{array}\right]\]
\[\mathbf{P}_h^{new}=\mathbf{R}(\theta)\mathbf{P}_h(\mathbf{R}(\theta))^T;\mathbf{P}_c^{new}=\mathbf{R}(\theta)\mathbf{P}_c(\mathbf{R}(\theta))^T\]
All the other values of $\theta$ will yield an $\mathbf{R}$ such that the pair $(\mathbf{P}_c^{new}\tilde{A}+\mathbf{P}_h^{new},-\mathbf{P}_h)$ is observable. Hence, one can arbitrarily place the poles of the system. We  select the gain $\mathbf{K}_2$ to place the poles at $z_{1,2}= \pm 0.1$ for two different values of $\theta$ that are randomly selected for comparison purposes as follows,
\[\theta=\frac{5\pi}{180}\mbox{ or } \frac{85\pi}{180}\rightarrow\]
\[ \mathbf{P}_c^{new}\tilde{A}+\mathbf{P}_h^{new}=\left[\begin{array}{cc} 1.41 &0.12\\0.25& 1.07\end{array}\right], \mathbf{K}_2=\left[\begin{array}{cc}20.09 & -40.18 \\11.29 & -22.58\end{array}\right]\]
and
\[\theta=\frac{45\pi}{180} \rightarrow\]
\[\mathbf{P}_h^{new}\tilde{A}+\mathbf{P}_h^{new}=\left[\begin{array}{cc} 1.20 &-0.15\\0.60& 0.55\end{array}\right],\mathbf{K}_2=\left[\begin{array}{cc}1.95 & -3.90 \\1.85 & -3.7\end{array}\right]\]

The closer $\theta$ is to $\frac{q\pi}{2},q \in \mathbb{Z}$ implies that a higher gain is required. This is an important consideration as it may lead to robustness issues when the system is subject to disturbances and noise. Using Theorem \ref{thm:nmp-filter}, the inverse filter for the above system when $\theta =\frac{45\pi}{180}$ is given by,
\begin{equation}\label{eq:invsys1}
  S^{inv} : \left\lbrace \begin{array}{l} \hat{\eta}(k+1)=\left(\begin{array}{cc} 3.15& -4.05 \\
  2.45&-3.15\end{array}\right)\hat{\eta}(k)+\\ \left[\begin{array}{cccc}-0.25&0.12&0.25&0\\
    -0.75&0.37&0.75&0\end{array}\right]\mathcal{U}(k-2M) \\
  \hat{\mathbf{U}}_{}(k)=\left(\begin{array}{cc}  1& 0 \\
   1&1\end{array}\right)\hat{\eta}(k)+\mathbf{U}^{aux}_{}(k-2M)\\
   \hat{u}(k)=\mathbf{I}_p \hat{\mathbf{U}}_{}(k)
   \end{array} \right.
\end{equation}
where $M=n$, $\mathcal{U}(k-2M)$ is defined by equation (\ref{eq:u-mathcal}) and $\mathbf{U}_{}^{aux}(k-2M)$ is given by,
\[\mathbf{U}_{}^{aux}(k-2M)=\left[\begin{array}{cc} 0.2308 &-0.4615\\-0.1538& 0.3077\end{array}\right]\mathbf{Y}_{}(k-2M)\]
Figure \ref{fig:sim1} shows the performance of the input inversion estimation filter corresponding to both values of $\mathbf{K}_2$.
\begin{figure}
  \centering
  \includegraphics[width=0.45\textwidth]{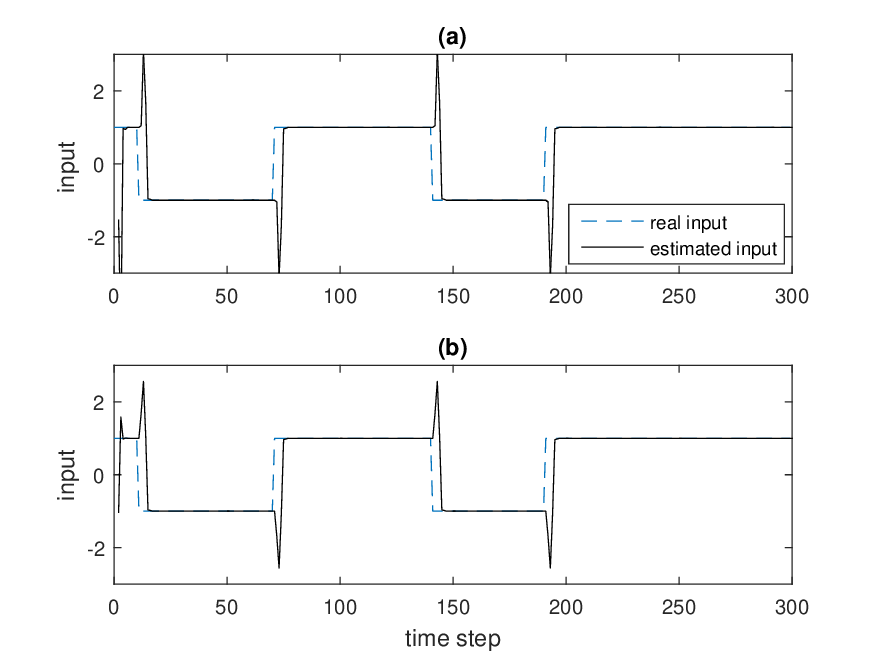}
  \caption{Input estimation for the system (\ref{sim:sys1}) using two different rotation matrices: (a) $\theta=\frac{5\pi}{180}$, and (b) $\theta=\frac{45\pi}{180}$.}\label{fig:sim1}
\end{figure}

{For the \textit{second simulation case study}}, we consider a non-minimum phase MIMO system that is governed by,
\begin{equation}\label{eq:sys3}
  S : \left\lbrace \begin{array}{l} x(k+1)=\left(\begin{array}{cccc}0& 0& 0&0.10 \\
  1&0& 0&-0.09 \\ 0&1&0& 0.28\\
  0&0&1&0.07 \end{array}\right)x(k)+\\ \left(\begin{array}{cc} 1&-0.80 \\
  0&-2.05 \\ 0& 5.13\\
  0&1.78 \end{array}\right)f(k)+w(k) \\
  y(k)=\left(\begin{array}{cccc}  -0.46& -0.35& -0.1& 0.14 \\
   0.59&-0.52&-0.01&0.04\end{array}\right)x(k)+v(k)
   \end{array} \right.
\end{equation}
The above system has two transmission zeros at $z_{1,2}=(-1.48,0.45)$. The system is subjected to \textit{both} a step and a ramp loss of effectiveness (LOE) faults in the channels 1 and 2, respectively. A random rotation matrix ($\mathbf{R} \in \mathbb{R}^{16 \times 16}$)\footnote{$\mathbf{R} \in \mathbb{R}^{2Ml\times 2Ml}$, $M=n=4$, and $l=2$.} is generated. The gain matrix $\mathbf{K_2} \in \mathbb{R}^{16 \times 16}$ is chosen such that $16$ poles of the filter (\ref{eq:nmp-filter-rmp}) are placed between $-0.1$ and $0.1$. The fault estimation results are shown in Figure \ref{fig:sim2}, which demonstrates the merits and capabilities of our proposed scheme for fault estimation of non-minimum phase systems.
\begin{figure}
  \centering
  \includegraphics[width=0.45\textwidth]{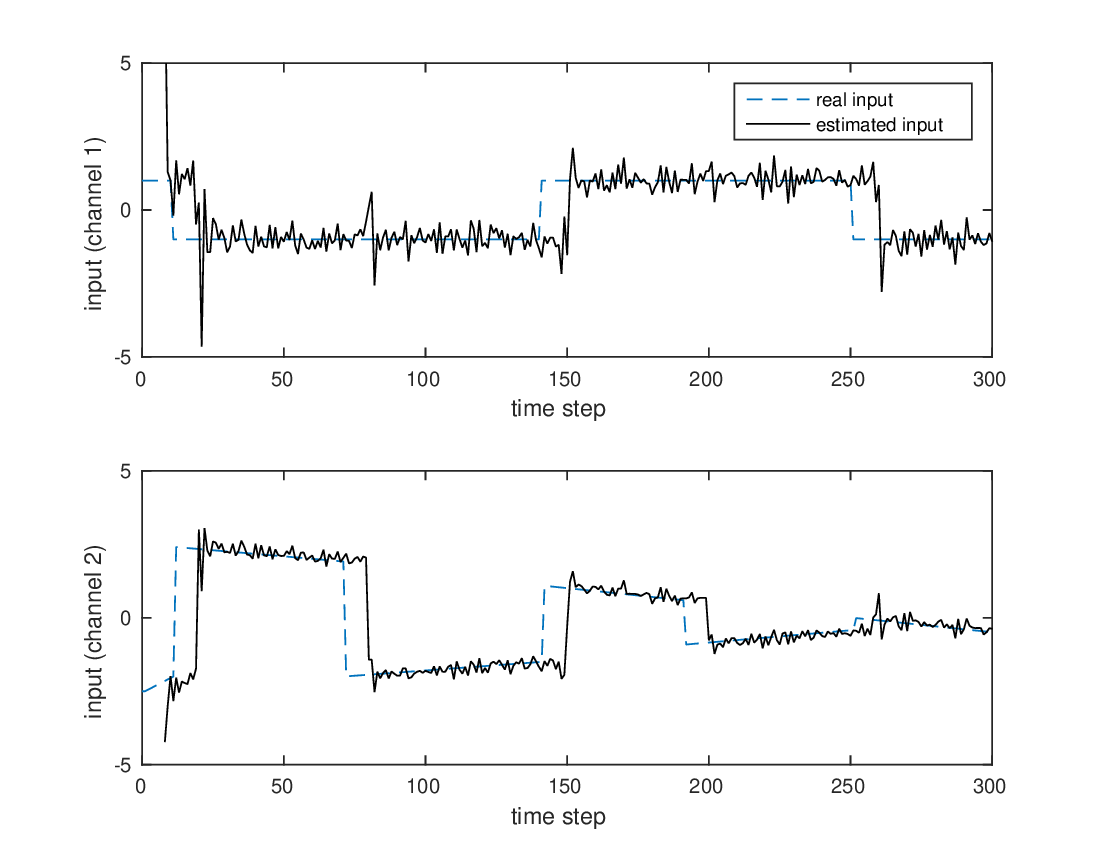}
  \caption{{The LOE fault input estimation of the MIMO non-minimum phase system (\ref{eq:sys3}).}}\label{fig:sim2}
\end{figure}
The most important advantage of our proposed solution arises as a result of the fact that it can handle systems with transmission zeros everywhere on the unit circle except at $z=1$.

In order to demonstrate the above point, consider the following {\textit{third simulation case study} } of the fault-to-output system,
\begin{equation}\label{eq:sysu}
\frac{Y(z)}{F(z)}=\frac{(z+1)(z^2+1)}{z^4}
\end{equation}
The simulation results for input estimation of this system are shown in Figure \ref{fig:simu}. The rotation matrix for constructing the filter (\ref{eq:nmp-filter}) is randomly generated. The gain matrix $\mathbf{K}_2$ is chosen such that the poles of the filter (\ref{eq:nmp-filter}) are placed at $z_1,\ldots,z_8=\pm 0.5, \pm 0.3571, \pm 0.2143, \pm 0.0714$. As can be seen from Figure \ref{fig:simu}, our proposed solution can successfully reconstructs the unknown fault even if the system has several transmission zeros on the unit circle.

Finally, for the { \textit{fourth simulation case study } } and as a comparative study, consider a MIMO system that is taken from the reference \cite{Marro2010815} with $A \in \mathbb{R}^{4 \times 4}$, $B \in \mathbb{R}^{4 \times 2}$ and $C \in \mathbb{R}^{2 \times 4}$ as follows,
\begin{equation}\label{eq:zatt-sys}
\left\lbrace\begin{array}{l} x(k+1)=\left[\begin{array}{cccc}0.6 &-0.3& 0&0 \\0.1 &1 &0 &0\\ -0.4& -1.5& 0.4& -0.3\\ 0.3 &1.1& 0.2&0.9\end{array}\right]x(k)+\\ \left[\begin{array}{cc}0 &0.4\\0&0 \\ 0&-0.1 \\0.1 &0.1 \end{array}\right]u(k)+w(k) \\
y(k)=\left[\begin{array}{cccc}1& 2& 3&4\\ 2& 1& 5& 6\end{array}\right]x(k)+v(k)
\end{array}\right.
\end{equation}
The system (\ref{eq:zatt-sys}) has two zeros at $z_1=0.6072$ and $z_1=1.9928$. The authors of \cite{Marro2010815} proposed a geometric approach and applied it to the system (\ref{eq:zatt-sys}) to achieve an \textit{almost} perfect estimation of the states and unknown inputs with a delay of 20 time steps ($n_d=20$). For comparison, our simulation results for the same example is shown in Figure \ref{fig:zatt}, which demonstrate that by using our proposed methodology the unknown inputs are almost perfectly reconstructed with only a delay of $n_d=8$. It should be noted that the approach that is proposed in \cite{Marro2010815} can handle any type of unknown input, whereas our approach is limited to step and ramp unknown inputs which covers a wide range of faults that occur in physical systems. The main advantage of our proposed methodology over the geometric approach that is proposed in \cite{Marro2010815} is the fact that it can handle systems with transmission zeros on the unit circle, whereas the approach in \cite{Marro2010815} \textit{cannot} handle this situation.

\begin{figure}
  \centering
  \includegraphics[width=0.45\textwidth]{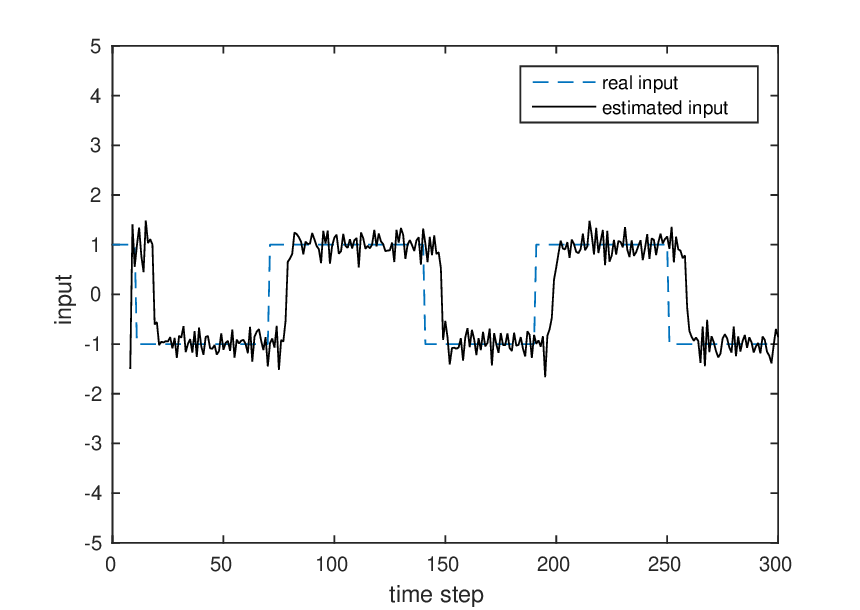}
  \caption{{The LOE fault estimation for the system (\ref{eq:sysu}). Noise is also injected to the input of the system.}}\label{fig:simu}
\end{figure}

\begin{figure}
  \centering
  \includegraphics[width=0.45\textwidth]{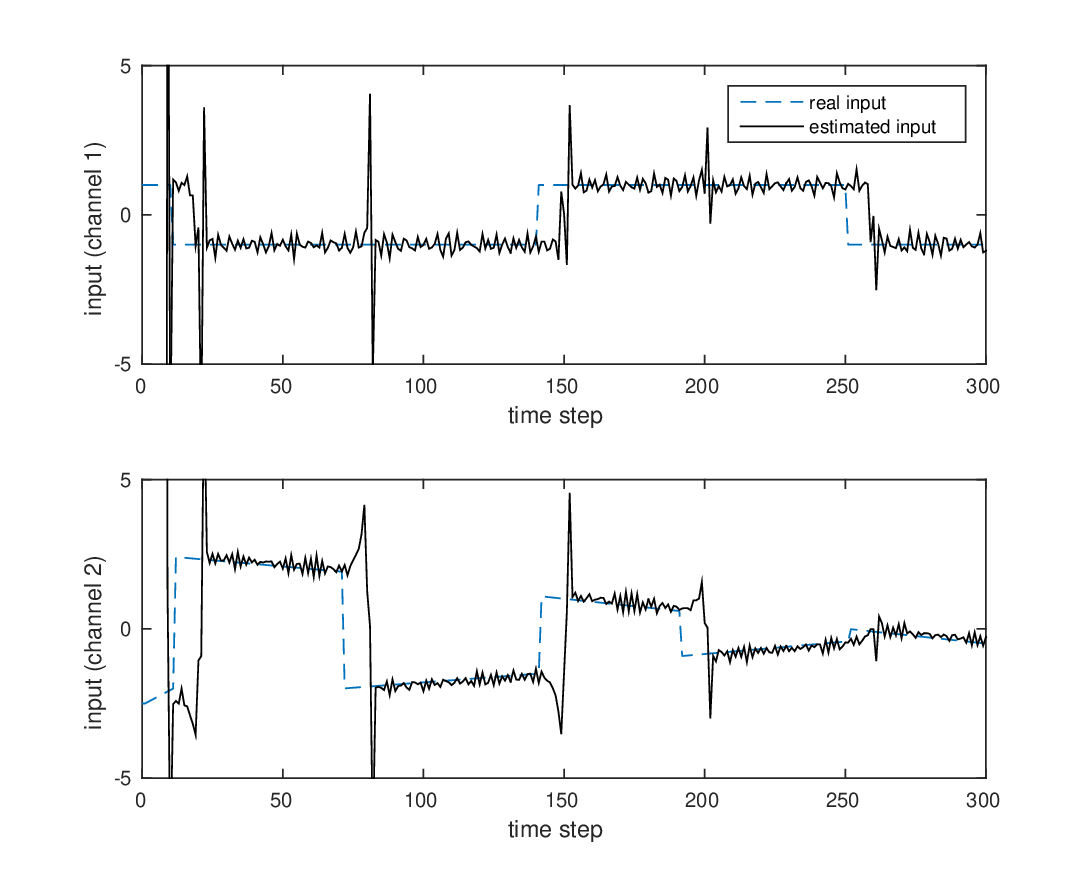}
  \caption{{Input estimation for the system (\ref{eq:zatt-sys}).}}\label{fig:zatt}
\end{figure}
\section{Conclusion}
We have developed an inversion-based fault estimation scheme for linear discrete-time systems. It was shown that our scheme yields an unbiased estimate of certain types of faults even if the fault-to-output dynamics has transmission zeros outside or on the unit circle (except at $z=1$). This is achieved by introducing a feedback that not only stabilizes the inverse dynamics (except those having transmission zeros at $z=1$), but also it provides an unbiased tracking of the unknown input. We have discussed the properties of the proposed inverse filter and conditions that are required for its stable design. We have also provided several illustrative simulation case studies that demonstrate the capabilities of our proposed methodologies. Yet, further research are required to generalize our proposed approach to a broader categories of faults.

\bibliographystyle{ieeetr}
\bibliography{engine}


\section{Appendix: Proof of Lemma \ref{lm:square-d-rank}}\label{app:square-d-rank}

If $D$ is full rank, then $rank(\mathbf{D}_{})$ is obviously greater than $2Mm-n$. If $D$ is zero or rank deficient, since the system has no transmission zeros, then at least $CA^{n-1}B$ is full rank. Note that $CA^{n-1}B$ does not appear in $\mathbf{D}_{}$ from the column $2Mm-n+1$ there after. Hence, it follows that $rank(\mathbf{D}_{})\geq 2Mm-n$.

For a square system, the equality also holds since one can express the measurement equation in the matrix format as follows,
\begin{equation}
\mathbb{E}\left(\mathbf{Y}_{}(k-2M)\right)=\mathbb{E}\left(\left[\begin{array}{cc}\mathbf{C}_{} & \mathbf{D}_{}\end{array}\right]\left[\begin{array}{c}x(k-2M) \\ \mathbf{U}_{}(k-2M) \end{array}\right]\right)
\end{equation}
Therefore, if $rank(\mathbf{D}_{})>2Mm-n$, certain columns of $\mathbf{C}_{}$ are linearly dependent with the columns of $\mathbf{D}_{}$, which implies that there exist a nonzero initial $x(k-2M)$ and a nonzero input sequence that will yield a zero output. This results in a contradiction, and therefore the rank condition should be satisfied.
\section{Appendix: Proof of Theorem \ref{thm:u-uaux}}\label{app:u-uaux}
If we subtract equation (\ref{eq:y-um}) from the measurement equation of the system $\mathbf{S}^{aug}$ and rewrite it in a matrix format, we will obtain,
{
\begin{equation}
\mathbb{E}\left(\left[\begin{array}{cc}\mathbf{C}_{} & \mathbf{D}_{}\end{array}\right]\left[\begin{array}{c}x(k-2M)-z(k-2M) \\ \mathbf{U}_{}(k-2M)-\mathbf{U}_{}^{aux}(k-2M) \end{array}\right]\right)=0
\end{equation}}
Since the system $\mathbf{S}$ does not have any transmission zeros, the columns of $\mathbf{C}_{}$ and $\mathbf{D}_{}$ are linearly independent. Hence, $\mathbf{C}_{}(x(k-2M)-z(k-2M))=0$, and
\begin{equation}\label{eq:hd-u-uaux}
{\mathbb{E}\left(\mathbf{D}_{}(\mathbf{U}_{}(k-2M)-\mathbf{U}_{}^{aux}(k-2M))\right)=0}
\end{equation}
On the other hand, from Lemma \ref{lm:square-d-rank}, it follows that the first $2Mm-n$ columns of $\mathbf{D}_{}$ must be linearly independent. Therefore, one can transform equation (\ref{eq:hd-u-uaux}) into the following format using basic operations on the last $n$ columns of $\mathbf{D}_{}$ and the last $n$ rows of {$\mathbb{E}\left(\mathbf{U}_{}^{aux}(k-2M)-\mathbf{U}_{}(k-2M)\right)$}. Specifically, we have,
{
\begin{equation}
\mathbb{E}\left(\left[\begin{array}{cc}\mathcal{D} & 0\end{array}\right]\left[\begin{array}{c}\mathcal{U}_{}^{aux}(k-2M)-\mathcal{U}_{}(k-2M) \\ \mathcal{X}\end{array}\right]\right)=0
\end{equation}}
where $\mathcal{D}$ is a nonsingular matrix that has the first $2Mm-n$ columns of $\mathbf{D}_{}$ and $\mathcal{U}_{}^{aux}(k-2M)-\mathcal{U}_{}(k-2M)$ is the first $2Mm-n$ rows of $\mathbf{U}_{}^{aux}(k-2M)-\mathbf{U}_{}(k-2M)$. Therefore, the first $2Mm-n$ rows of {$\mathbb{E}\left(\mathbf{U}_{}^{aux}(k-2M)- \mathbf{U}_{}(k-2M)\right)$} are zero as stated. This completes the proof of the theorem.
\section{Appendix: Proof of Theorem \ref{thm:zeros-of-a-bipkc}}\label{app:zeros-of-a-bipkc}

Note that the eigenvalues of $A-B\mathbf{I}_p\mathbf{D}_{}^+\mathbf{C}_{}$ are obtained by solving,
\begin{equation}\label{eq:eig-org}
|z\mathbf{I}-A+B\mathbf{I}_p\mathbf{D}_{}^+\mathbf{C}_{}|=0
\end{equation}
If the system is square, then $\mathbf{D}_{}^+$ is a nonzero square matrix. Therefore, one can equivalently solve the equation,
\begin{equation}
\left|\mathbf{D}_{}^+\right||z\mathbf{I}-A+B\mathbf{I}_p\mathbf{D}_{}^+\mathbf{C}_{}|=0
\end{equation}
On the other hand, using the Schur identity, we have,
\begin{equation}\label{eq:schur}
\left|\mathbf{D}_{}^+\right||z\mathbf{I}-A+B\mathbf{I}_p\mathbf{D}_{}^+\mathbf{C}_{}| = \left|\left[\begin{array}{cc}
z\mathbf{I}-A &-B\mathbf{I}_p \\ \mathbf{C}_{} & \mathbf{D}_{} \end{array}\right]\right|
\end{equation}
Let us partition the terms $\mathbf{C}_{}$ and  $\mathbf{D}_{}$ as follows,
\begin{equation}
\mathbf{C}_{} = \left(\begin{array}{c} C \\ \hline \mathcal{C}^- \end{array} \right)=\left(\begin{array}{c}C\\ \hline CA\\ \vdots \\ CA^{2M-1}\end{array}\right)
\end{equation}
\begin{equation}
\small{
\mathbf{D}_{}=\left(\begin{array}{c|c}D & 0 \\ \hline \mathcal{D}^-_{21} & \mathcal{D}^-_{22}\end{array}\right) =\left(\begin{array}{c|ccc} D&0& \ldots & 0\\ \hline CB& D& \ldots &0 \\ \vdots & \vdots & \vdots & \vdots \\ CA^{2M-1}B&CA^{2M-2}B&\ldots & D\end{array}\right)}
\end{equation}
It now follows that the right-hand side of equation (\ref{eq:schur}) can be partitioned as,
\begin{equation}
\left[\begin{array}{cc}
z\mathbf{I}-A &-B\mathbf{I}_p \\ \mathbf{C}_{} & \mathbf{D}_{} \end{array}\right]= \left[\begin{array}{cc|c}
z\mathbf{I}-A &-B & 0\\ C & D & 0 \\ \hline \mathcal{C}^- & \mathcal{D}^-_{21} &\mathcal{D}^-_{22} \end{array}\right]
\end{equation}
Thus, if $\mathcal{D}^-_{22}$ is full row rank, according to the Schur identity, equation (\ref{eq:eig-org}) has only one set of solutions that are given by,
\begin{equation}
\left|\left[\begin{array}{cc}
z\mathbf{I}-A &-B\mathbf{I}_p \\ \mathbf{C}_{} & \mathbf{D}_{} \end{array}\right]\right|=0
\end{equation}
and these are exactly the transmission zeros of the system $\mathbf{S}$. However, if $\mathcal{D}^-_{22}$ is rank deficient, then certain rows of $\left[\begin{array}{ccc}\mathcal{C}^- & \mathcal{D}^-_{21} &\mathcal{D}^-_{22}\end{array}\right]$ are linearly dependent with the rows of $\left[\begin{array}{ccc}-A& -B &0\end{array}\right]$. Hence, $z=0$ is also a solution. On the other hand, since equation (\ref{eq:eig-org}) must have $n$ eigenvalues, therefore if the system $\mathbf{S}$ has $p$ transmission zeros, then $z=0$ is a solution of multiplicity $n-p$, and this completes the proof of the theorem.

\section{Appendix: Proof of Lemma \ref{lm:ipd}}\label{app:lm-ipd}
Note that the first $m$ columns of $\mathbf{D}_{}$ are linearly independent. Therefore,
\[rank(\left[\begin{array}{c}\mathbf{I}_p \\ \mathbf{D}_{}\end{array}\right])=rank(\mathbf{D}_{})\]
which implies that the subspace spanned by the rows of $\mathbf{I}_p$ belongs to the row space that is spanned by the rows of $\mathbf{D}_{}$. Therefore, $\mathbf{I}_P.\mathcal{N}(\mathbf{D}_{})=0$. This completes the proof of the lemma.
\section{Appendix: Proof of Theorem \ref{thm:mp-filter}}\label{app:mp-filter}
First, we show that {$\mathbb{E}(\hat{e}(k)-e(k)) \rightarrow 0$ as $k \rightarrow \infty$}. Then we show this will yield {$\mathbb{E}(\hat{u}(k)-u(k-2M)) \rightarrow 0$ as $k \rightarrow \infty$}. Note that the governing dynamics of $e(k)$ is given by equation (20). Therefore, in view of equations (\ref{eq:state-error-dyn}) and (\ref{eq:mp-filter}) we have, {$\hat{e}(k+1)-e(k+1)=(A-B\mathbf{I}_p\mathbf{D}_{}^+\mathbf{C}_{})(\hat{e}(k)-e(k))+\mathcal{ST}$, where $\mathcal{ST} = -\mathbf{D}^+(\mathbf{G}\mathbf{W}(k-2M)+\mathbf{V}(k-2M)-w(k-2M))$}. Since the system $\mathbf{S}$ is minimum phase, therefore, according to Theorem \ref{thm:zeros-of-a-bipkc}, {$\mathbb{E}(\hat{e}(k)-e(k)) \rightarrow 0$ as $k \rightarrow \infty$} (note that Theorem \ref{thm:zeros-of-a-bipkc} implies that $A-B\mathbf{I}_p\mathbf{D}_{}^+\mathbf{C}_{}$ is Hurwitz if the system $\mathbf{S}$ is minimum phase). Note that the error in the unknown input reconstruction is given by,
{$\mathbb{E}(\hat{\mathbf{U}}_{}(k)-\mathbf{U}_{}(k-2M))
=\mathbb{E}(-\mathbf{D}_{}^+\mathbf{C}_{}\hat{e}(k)+\mathbf{U}^{aux}_{}(k-2M)-\mathbf{U}_{}(k-2M))$ $
\rightarrow$ $\mathbb{E}( -\mathbf{D}_{}^+\mathbf{C}_{}e(k)-\delta \mathbf{U}_{}(k))
= \mathbf{D}_{}^+\mathbf{D}_{} \delta \mathbf{U}_{}(k)-\delta \mathbf{U}_{}(k)$}.
Consequently, we have,
\begin{equation}\label{eq:tmp1}
{\mathbb{E}(\hat{u}(k)-u(k-2M)) \rightarrow \mathbf{I}_p(\mathbf{D}_{}^+\mathbf{D}_{} -\mathbf{I})\delta \mathbf{U}_{}(k)}
\end{equation}
where $(\mathbf{D}_{}^+\mathbf{D}_{} -\mathbf{I})$ is the projector onto the null space of $\mathbf{D}_{}$. Since $\mathbf{I}_P .\mathcal{N}(\mathbf{D}_{})=0$, according to Lemma \ref{lm:ipd}, the right-hand side of equation (\ref{eq:tmp1}) is zero. {Therefore, it follows that $\mathbb{E}(\hat{u}(k)) \rightarrow \mathbb{E} (u(k-2M))$ as $k \rightarrow \infty$}.

\section{Appendix: Proof of Lemma \ref{rm:obs3}}\label{app:obs3}
We use the Hautus test (\cite{sontag2013mathematical}) to show this lemma. The observability matrix of the pair $(\mathbf{P}_c^{new}\tilde{A}+\mathbf{P}_h^{new},-\mathbf{P}_h)$ is equivalent to the controllability matrix of the pair $((\tilde{A}+\mathbf{P}_h^{new})^T,-(\mathbf{P}_h)^T)$. The pair $((\mathbf{P}_c^{new}\tilde{A}+\mathbf{P}_h^{new})^T,-(\mathbf{P}_h)^T)$ is controllable if
\[rank\left(\left[\begin{array}{cc}(\mathbf{P}_c^{new}\tilde{A}+\mathbf{P}_h^{new})^T-\lambda \mathbf{I}&-(\mathbf{P}_h)^T\end{array}\right]\right)=2Ml\]
for all $\lambda \in \mathbb{C}$. We now show that when the square system $\mathbf{S}$ has a transmission zero equal to 1, then this condition is not satisfied for $\lambda=1$. Equivalently, there exists a nonzero $w$ such that $w\Theta=0$, for $\lambda=1$, where
$\Theta=\left[\begin{array}{cc}(\mathbf{P}_c^{new}\tilde{A}+\mathbf{P}_h^{new})^T-\lambda \mathbf{I}&-(\mathbf{P}_h)^T\end{array}\right]$.
When $\lambda=1$, it follows that,
\begin{eqnarray}
(\mathbf{P}_c^{new}\tilde{A}+\mathbf{P}_h^{new})^T-\lambda \mathbf{I}&=&(\mathbf{P}_c^{new}\tilde{A}+\mathbf{P}_h^{new})^T- \mathbf{I} \nonumber \\
&=&(\mathbf{P}_c^{new}\tilde{A}-\mathbf{P}_c^{new})^T
\end{eqnarray}
Recall from Theorem \ref{thm:zeros-of-a-bipkc} that the transmission zeros of $\mathbf{S}$ are the eigenvalues of $A-B\mathbf{I}_p\mathbf{D}_{}^+\mathbf{C}_{}$. Hence, if the system $\mathbf{S}$ has a transmission zero equal to 1, there exists a nonzero $v$ such that $\mathbf{P}_c^{new}\tilde{A}v=\mathbf{P}_c^{new}v$. Therefore, by selecting
$w=\left[\begin{array}{cc} v^T& 0\end{array}\right]$,
one can achieve $w\Theta=0$ independent of the choice of the rotation matrix $\mathbf{R}$. This completes the proof of the lemma.
\section{Appendix: Proof of Theorem \ref{thm:nmp-filter}}\label{app:nmp-filter}
{First, it is shown that $\mathbb{E}(\hat{\eta}(k)-\eta(k)) \rightarrow 0$ as $k \rightarrow \infty$. Then, we show that it follows that $\mathbb{E}(\hat{u}(k)-u(k-2M)) \rightarrow 0$ as $k \rightarrow \infty$. If one subtracts equation (\ref{eq:eta-hat}) from the equation (\ref{eq:eta-dyn}), one will have,}
{\begin{eqnarray}\label{eq:dum-eq1}
\hat{\eta}(k+1)&-&\eta(k+1)=(\tilde{A}+\mathbf{P}_h^{new}+\mathbf{K}_2\mathbf{P}_h)(\hat{\eta}(k)-\eta(k)) \nonumber \\ &+&(\mathbf{P}_h^{new}\tilde{A}-\mathbf{P}_h^{new})\eta(k)+\mathbf{P}_h^{new}\mathbf{C}_{}\mathcal{U}(k)+ \mathcal{ST}\nonumber \\
&=& (\tilde{A}+\mathbf{P}_h^{new}+\mathbf{K}_2\mathbf{P}_h)(\hat{\eta}(k)-\eta(k)) \nonumber \\ &+&\mathbf{P}_h^{new}(\eta(k+1)-\eta(k))+\mathcal{ST'}
\end{eqnarray}}
where,
\[ {\mathcal{ST'}=\mathbf{C}\mathbf{D}^+(\mathbf{G}\mathbf{W}(k-2M)-\mathbf{V}(k-2M))}\]
{Let us define $e(k)=\mathbb{E}(\hat{\eta}(k)-\eta(k))$. Also let us take the $\mathcal{Z}$-transform of both sides of equation (\ref{eq:dum-eq1}), which after some rearrangements gives us,}
\begin{equation}
e(z)=(zI-\tilde{A}-\mathbf{P}_h^{new}-\mathbf{K}_2\mathbf{P}_h)^{-1}\mathbf{P}_h^{new}(z-1)\eta(z)
\end{equation}
{If the input to the system $\mathbf{S}$ is a step function, then according to the final value theorem, we have
$\lim_{k\rightarrow\infty}e(k)=\lim_{z\rightarrow1}(z-1)e(z)=0$,
which implies that $\mathbb{E}(\hat{\eta}(k)-\eta(k)) \rightarrow 0$ as $k \rightarrow \infty$.}

The estimation error in the unknown input reconstruction is given by,
\begin{eqnarray}
\hat{\mathbf{U}}_{}(k)&-&\mathbf{U}_{}(k-2M)=\mathbf{D}_{}^+\hat{\eta}(k)\nonumber \\ &+&\mathbf{U}^{aux}_{}(k-2M)-\mathbf{U}_{}(k-2M) \nonumber \\
&\rightarrow& \mathbf{D}_{}^+\eta(k)-\delta \mathbf{U}_{}(k) \nonumber \\
&=& \mathbf{D}_{}^+\mathbf{D}_{} \delta \mathbf{U}_{}(k)-\delta \mathbf{U}_{}(k)
\end{eqnarray}
Thus, we have,
{
\begin{equation}\label{eq:tmp2}
\mathbb{E}(\hat{u}(k)-u(k-2M)) \rightarrow \mathbf{I}_p(\mathbf{D}_{}^+\mathbf{D}_{} -\mathbf{I})\delta \mathbf{U}_{}(k)
\end{equation}}
where $(\mathbf{D}_{}^+\mathbf{D}_{} -\mathbf{I})$ is the projector onto the null space of $\mathbf{D}_{}$. Since $\mathbf{I}_P .\mathcal{N}(\mathbf{D}_{})=0$, according to Lemma \ref{lm:ipd}, the right-hand side of equation (\ref{eq:tmp1}) is zero. Therefore, it can be concluded that {$\mathbb{E}(\hat{u}(k)) \rightarrow \mathbb{E}(u(k-2M))$ as $k \rightarrow \infty$}. This completes the proof of the theorem.
\section{Appendix: Proof of Theorem \ref{thm:nmp-filter-rmp}}\label{app:nmp-filter-rmp}
{First, it is shown that $\mathbb{E}(\hat{\eta}(k)-\eta(k)) \rightarrow 0$ as $k \rightarrow \infty$. Then we show that it follows that $\mathbb{E}(\hat{u}(k)-u(k-2M)) \rightarrow 0$ as $k \rightarrow \infty$}. Let us define the following dummy variables,
\[A_1=\mathbf{P}_h^{new}\tilde{A}^2-2\mathbf{P}_h^{new}\tilde{A}+\tilde{A}+\mathbf{P}_h^{new}\]
\[A_2=\mathbf{P}_h^{new}\tilde{A}^2-2\mathbf{P}_h^{new}\tilde{A}+\tilde{A}+\mathbf{P}_h^{new}+\mathbf{K}_2 \mathbf{P}_h\]
\[B_1=\mathbf{P}_h^{new}\mathbf{C}_{}B_F\mathcal{U}(k-2M+1)\]
\[+\left(\mathbf{P}_h^{new}\tilde{A}-2\mathbf{P}_h^{new}+\mathbf{I}\right)\mathbf{C}_{}B_F\mathcal{U}(k-2M)\]
\[B_2=\mathbf{P}_h^{new}\mathbf{C}_{}B_F\mathcal{U}(k-2M+1)\]
\[+\left(\mathbf{P}_h^{new}\tilde{A}-2\mathbf{P}_h^{new}\right)\mathbf{C}_{}B_F\mathcal{U}(k-2M)\]
If we subtract the state equation of the filter (\ref{eq:nmp-filter-rmp}) from that of equation (\ref{eq:eta-dyn}), we will have
$\hat{\eta}(k+1)-\eta(k+1)=\tilde{A}\eta(k)+\mathbf{C}_{}B_F\mathcal{U}(k)-A_2\hat{\eta}(k)- B_1+\mathcal{ST}
= A_2(\hat{\eta}(k)-\eta(k))-A_1\eta(k)-B_2 +\mathcal{ST} = A_2(\hat{\eta}(k)-\eta(k))-\mathbf{P}_h^{new}(\eta(k+2)-2\eta(k+1)+\eta(k)) +\mathcal{ST}$.
Let us define as before {$e(k)=\mathbb{E}(\hat{\eta}(k)-\eta(k))$}. Also, let us take the $\mathcal{Z}$-transform of both sides of the equation which after some rearrangements gives,
$e(z)=-(zI-A_2)^{-1}\mathbf{P}_h^{new}(z-1)^2\eta(z)$.
If the input is a step or a ramp function, then according to the final value theorem it follows that,
$\lim_{k\rightarrow\infty}e(k)=\lim_{z\rightarrow1}(z-1)e(z)=0$,
which implies that {$\mathbb{E}(\hat{\eta}(k)-\eta(k)) \rightarrow 0$ as $k \rightarrow \infty$}. The remainder of the proof follows along similar lines as those invoked in the proof of Theorem 4 in Appendix \ref{app:nmp-filter}, and therefore these details are omitted for brevity.
\end{document}